\documentclass[aps,prl
,showpacs
,twocolumn
,superscriptaddress
]
{revtex4-1}
\usepackage{amsmath,amsthm,amssymb}
\usepackage{array}
\usepackage{graphicx}
\usepackage{subfig}
\usepackage{fancyhdr}
\usepackage{color}
\usepackage{extarrows}
\usepackage{enumerate}
\newtheorem{definition}{Definition}
\newtheorem{theorem}{Theorem}
\newtheorem{lemma}{Lemma}

\newtheorem{example}{Example}
\begin{document}

\title{Constructive Proof of Global Lyapunov Function as Potential Function}

\author{Ruoshi Yuan}
\author{Yian Ma}
\author{Bo Yuan}
\email{Corresponding author. Email: boyuan@sjtu.edu.cn}
\affiliation{Department of Computer Science and Engineering}
\author{Ping Ao}
\email{Corresponding author. Email: aoping@sjtu.edu.cn}
\affiliation{Shanghai Center for Systems Biomedicine and Department of Physics\\Shanghai Jiao Tong University, Shanghai, 200240, China}
\date{\today}

\begin{abstract}
We provide a constructive proof on the equivalence of two fundamental concepts: the global Lyapunov function in engineering and the potential function in physics, establishing a bridge between these distinct fields.
 This result suggests new approaches on the significant unsolved problem namely to construct Lyapunov functions for general nonlinear systems through the analogy with existing methods on potential functions.
 In addition, we show another connection that the Lyapunov equation is a reduced form of the generalized Einstein relation for linear systems.
\end{abstract}
\maketitle

The classical Lyapunov function \cite{lyapunov1992general,hirsch1974differential,sastry1999nonlinear,haddad2008nonlinear} has been widely applied in engineering for stability analysis though restricted for a specified fixed point.
Constructing such functions for general nonlinear systems is of great theoretical and practical interests \cite{johansson1998computation,papachristodoulou2002construction}, but still a challenge that would require one's ``divine inspiration" in application \cite{strogatz2000nonlinear}.
We introduce in this letter a global Lyapunov function as the natural generalization of its classical local version. Inside a neighborhood of a stable fixed point, the global Lyapunov function usually reduces to a classical one. Moreover, it enables us to do quantitative analysis of complex dynamical behaviors far from equilibrium (e.g.~multi-stable states and periodic attractors) which are ubiquitous in real systems \cite{scheffer2001catastrophic,field1974oscillations} but beyond the scope of the classical one.
Research on a fundamental concept in physics, the potential function, has been motivated recently by uncovering global principles of complex dynamics in biology \cite{frauenfelder1991energy,wang2008potential,ao2009global,qian2010}, physics \cite{ge2009thermodynamic,wang2003energy} and control theory \cite{maschke2002energy,wang2003generalized}.
One of the present authors has proposed a general construction of potential functions for stochastic dynamics \cite{ao2004potential,ao2008emerging}.
It was first formulated during the study of the robustness of a stochastic switch \cite{zhu2004calculating}.
Explicit results for fixed point \cite{kwon2005structure} and limit cycle systems \cite{zhu2006limit} have also been derived.
We demonstrate here that the global Lyapunov function is actually the potential function, which connects engineering to physics and indicates systematic approaches for constructing Lyapunov functions.

To avoid unnecessary mathematical complication, here we will only consider smooth dynamics. The present results can be directly extended to more general systems.
The definition of the classical Lyapunov function for a smooth system
\begin{align}
\dot{\mathbf{q}}=\mathbf{f}(\mathbf{q})\,,
\end{align}
where $\mathbf{f}:\mathbb{R}^n\xrightarrow{}\mathbb{R}^n$ is given by \cite{hirsch1974differential}
\begin{definition}[Lyapunov Function]
\label{def:Lyapunov}
Let $\mathbf{q}^*$ be a fixed point for the system and $L: \mathcal{O}\to\mathbb{R}$ a $C^{1}$ function defined on an open set $\mathcal{O}$ containing $\mathbf{q}^*$. Then $L$ satisfying the following conditions is called a Lyapunov function.
\begin{enumerate}
\item[($a$)] $L(\mathbf{q}^*) = 0$ and $L(\mathbf{q}) > 0$ if $\mathbf{q}\neq \mathbf{q}^*$ \label{item_l:a};
\item[($b$)] вл$\dot{L}(\mathbf{q})=\frac{dL}{dt}|_\mathbf{q}\leqslant0$ for all $\mathbf{q}\in\mathcal{O}$\label{item_l:b}.
\end{enumerate}
\end{definition}
($a$) implies that the fixed point $\mathbf{q}^*$ is a local extremum such that $\nabla L(\mathbf{q}^*)=0$. Although ($a$) is sufficient for the determination of stability, it seems too strict for a quantity with global meaning as well as applicable for complex dynamical behaviors.
We have to amend this condition in order to contain at least saddle points.
For such a purpose, we take the weaker form of ($a$): $\nabla L(\mathbf{q}^*)=0$ for all fixed points $\mathbf{q}^*$, to fit our generalization.
This together with ($b$) constitutes the definition of \textit{global Lyapunov function}.
\begin{definition}[Global Lyapunov Function]
\label{def:cand_Lyapunov}

  Let $\psi:\mathbb{R}^n\xrightarrow{}\mathbb{R}$ be a $C^1$ function. Then $\psi$ satisfying the following conditions is called a global Lyapunov function.
  \begin{enumerate}[(a)]

  \item $\nabla \psi(\mathbf{q^*})=0$ for all $\mathbf{q}^*$ where $\dot{\mathbf{q}}=\mathbf{f}(\mathbf{q}^*)=0$ \label{item_g:a};
  \item $\dot{\psi}(\mathbf{q})=\frac{d\psi}{dt}|_\mathbf{q}\leqslant0$ for all $\mathbf{q}\in \mathbb{R}^n$ \label{item_g:b}.
  \end{enumerate}
\end{definition}
The global Lyapunov function is equivalent to the potential function obtained from the following proposed Canonical Form.

The evolution of a deterministic dynamical system described by a set of differential equations
can be considered as a \textit{massless} particle moving along the trajectories inside the phase space. From a physical point of view, it is natural to explain the motion of this particle as a consequence of the underlying driving forces where
$\mathbf{F}_{driving}=m\ddot{\mathbf{q}}=0$.
These forces can be separated generally in physics into a dissipative and a conservative part
$\mathbf{F}_{driving}=\mathbf{F}_{conservative}+\mathbf{F}_{dissipative}=0$.
Without the loss of generality, we use a frictional force to represent the dissipative part
$\mathbf{F}_{dissipative}=-S\dot{\mathbf{q}}$ and a Lorentz force together with an energy induced force as the conservative part $\mathbf{F}_{conservative}=e\dot{\mathbf{q}}\times \mathbf{B}+\left[-\nabla \psi(\mathbf{q})\right]$, thus
$-S \dot{\mathbf{q}}+e\dot{\mathbf{q}}\times \mathbf{B}-\nabla \psi(\mathbf{q})=0$
where $S$ is symmetric and semi-positive definite.

The semi-positive definite requirement for $S$ guarantees the resistance of the frictional force whose valid values are restricted to the negative half space. The potential function $\psi$ here serves as an indicator demonstrating the influence of the other two forces onto the energy of the system. It is apparent that the energy induced force is equal and opposite everywhere to the resultant of the other forces as such
\begin{align}
\label{eq:intro_origin}
S\dot{\mathbf{q}}+e \mathbf{B}\times\dot{\mathbf{q}}=-\nabla \psi (\mathbf{q})\,.
\end{align}
The work done by the frictional force is then the reduced amount of $\psi$.

However, a problem occurs for systems whose dimension is higher than 3, since the cross product $\mathbf{B}\times\dot{\mathbf{q}}$ is undefined. In order to generalize \eqref{eq:intro_origin} to be valid for arbitrary n-dimensional systems, we introduce a generalized form of this vector-valued cross product $\mathbf{B}\times\dot{\mathbf{q}}$ as $T\dot{\mathbf{q}}$, where $T$ is an antisymmetric matrix. This definition is consistent with the three dimensional case, since $\mathbf{B}\times\dot{\mathbf{q}}=T\dot{\mathbf{q}}$ when $T_{ij}=-\varepsilon_{ijk}B_k$ and $\varepsilon_{ijk}$ being the Levi-Civita symbol.

Hence by setting $e=1$, we reach
\begin{align}
\label{eq:CFE}
\left[S+T\right] \dot{\mathbf{q}}=-\nabla \psi (\mathbf{q})\,,
\end{align}
where $S$ is symmetric and semi-positive definite and $T$ is antisymmetric. \eqref{eq:CFE} is
referred to as the \textit{Canonical Form} which induces $n(n-1)/2$ equations (utilizing the matrix-valued cross product in Definition \ref{def:cross_product}):
\begin{align}
\label{eq:D_CFE}
\nabla\times\left[\left(S+T\right)\dot{\mathbf{q}}\right]=0\,.
\end{align}

Symmetrically, it is proper to require $\psi$ to be convertible back to its original system. This gives the \textit{Standard Form}:
\begin{align}
\label{eq:SFE}
\dot{\mathbf{q}}=-\left[D+Q\right]\nabla \psi (\mathbf{q})\,,
\end{align}
where $D$ is symmetric and semi-positive definite, $Q$ is antisymmetric. The connection between $S$, $T$ and $D$ is characterized by the generalized Einstein relation (GER) \cite{ao2004potential,ao2008emerging}:
\begin{align}
\label{GER}
[S+T]D[S-T]=S\,,
\end{align}
This gives the other $n(n+1)/2$ equations. For a chosen $D$ with proper boundary conditions, we obtain $[S+T]$ ($n^2$ unknowns) by solving these $n^2$ equations provided by \eqref{eq:D_CFE} and \eqref{GER}, then $\psi$ can be derived from \eqref{eq:CFE}.

 Equation \eqref{GER} demonstrates the general relationship between friction and diffusion for stochastic dynamics \cite{ao2008emerging}. $D$ is the diffusion matrix indicating the random driving force. Deterministic system can be considered generally as a stochastic system with the noise being zero in strength.
 For deterministic dissipative system, there exists a frictional force that has the common origin with the \textit{undefined} (by its differential equations) random driving force.
 This point is implied by the fluctuation-dissipation theorem \cite{kubo1966fluctuation} or more generally the equation \eqref{GER}.
\begin{theorem}
\label{thm:equivalence}
For a certain dynamical system, any potential function $\psi$ obtained from \eqref{eq:CFE} is a global Lyapunov function. Conversely, explicit construction of $S$ and $T$ can be given for any global Lyapunov function of the system.
\end{theorem}
\begin{definition}[Matrix-valued Cross Product]
\label{def:cross_product}
The matrix-valued cross product of two vectors $\mathbf{x},\mathbf{y}\in\mathbb{R}^n$ is given by
$\mathbf{x} \times \mathbf{y}=A=(a_{ij})_{n\times n}=(x_iy_j-x_jy_i)_{n\times n}$.
The output is no longer a vector but an antisymmetric matrix.
\end{definition}

\begin{lemma}For arbitrary three vectors $\mathbf{x},\mathbf{y},\mathbf{z}\in\mathbb{R}^n$ (Note here the dot product and the matrix-valued cross product is precedent to the matrix multiplication),
\begin{equation}
\mathbf{x}\cdot\mathbf{y}\mathbf{z}=\mathbf{z}\cdot\mathbf{x}\mathbf{y}+\mathbf{z}\times\mathbf{y}\mathbf{x}\,.
\end{equation}
\label{lemma:cross_product}
The proof is straightforward since
$(\mathbf{z}\cdot\mathbf{x}\mathbf{y}+\mathbf{z}\times\mathbf{y}\mathbf{x})_{i}=\sum_kz_kx_ky_i+\sum_k(z_iy_k-z_ky_i)x_k
=(\mathbf{x}\cdot\mathbf{y}\mathbf{z})_i$.

\end{lemma}
\begin{proof}[Proof of Theorem \ref{thm:equivalence}]
From \eqref{eq:CFE}, $\dot{\mathbf{q}}=\mathbf{f}(\mathbf{q}^*)=0 \Rightarrow \nabla \psi (\mathbf{q}^*)=0$. Note that
$ \frac{d}{dt}{\psi}(\mathbf{q})=\dot{\mathbf{q}}^T\nabla \psi (\mathbf{q})
=-\dot{\mathbf{q}}^T \left[ S(\mathbf{q})+T(\mathbf{q}) \right] \dot{\mathbf{q}}
 =-\dot{\mathbf{q}}^T S(\mathbf{q}) \dot{\mathbf{q}}
\leqslant 0$, we find that $\psi$ satisfies $\dot{\psi}\leqslant0$ for all $\mathbf{q}\in \mathbb{R}^n$. Hence $\psi$ is a global Lyapunov function according to Definition \ref{def:cand_Lyapunov}.

Conversely, for any global Lyapunov function $\psi(\mathbf{q})$ of a given system $\dot{\mathbf{q}}=\mathbf{f}(\mathbf{q})$, by setting
\begin{align}
\label{exp_S}
S&=-\frac{\nabla \psi\cdot\mathbf{f}}{\mathbf{f}\cdot\mathbf{f}}E\,,\\
\label{exp_T}
T&=-\frac{\nabla \psi\times \mathbf{f}}{\mathbf{f}\cdot\mathbf{f}}\,,
\end{align}
  by utilizing $\nabla\psi\cdot\mathbf{f}=\nabla\psi\cdot\dot{\mathbf{q}}=\dot{\psi}\leqslant0$, $S$ is symmetric and semi-positive definite and $T$ is antisymmetric by the definition of the matrix-valued cross product.
 From Lemma \ref{lemma:cross_product}, we can obtain
  \begin{align}
   \mathbf{f}\cdot\mathbf{f}\nabla \psi=\nabla \psi\cdot\mathbf{f}\mathbf{f}+\nabla \psi\times \mathbf{f}\mathbf{f}\,,
  \end{align}
  by letting $\mathbf{x}=\mathbf{y}=\mathbf{f}$ and $\mathbf{z}=\nabla \psi$. Then
\begin{align}
\label{eq:guarantee}
\left[S+T\right]\dot{\mathbf{q}}=-\frac{\left(\nabla \psi\cdot\mathbf{f}E+\nabla \psi\times \mathbf{f}\right)\mathbf{f}}{\mathbf{f}\cdot\mathbf{f}}=-\nabla \psi\,.
\end{align}
When $\dot{\mathbf{q}}=0$, since $\nabla\psi=0$, the singularity of $S$ and $T$ will not affect $\psi$.
This demonstrates that any global Lyapunov function for a given system will satisfy \eqref{eq:CFE}.
\end{proof}

Explicit construction for the chosen $D$ and $Q$ fulfilling \eqref{eq:SFE} and \eqref{GER} can be provided under the former configuration of $S$ and $T$ utilizing the matrix-valued cross product,
\begin{align}
\label{eq:d}
D&=-\left[\frac{\mathbf{f}\cdot\mathbf{f}}{\nabla \psi \cdot \mathbf{f}}
E+\frac{\left(\nabla \psi\times \mathbf{f}\right)^2}{\left(\nabla \psi \cdot \mathbf{f}\right)\left(\nabla \psi\cdot\nabla \psi\right)}\right]\,,\\
\label{eq:q}
Q&=\frac{\nabla \psi\times \mathbf{f}}{\nabla \psi\cdot\nabla \psi}\,,
\end{align}
whose proof is straightforward by the Lemma \ref{lemma_cross3}.
\begin{lemma}[]
\label{lemma_cross3}
$\forall \mathbf{x}, \mathbf{y}\in \mathbb{R}^n:$
\begin{align}
\left(\mathbf{x}\times \mathbf{y}\right)^3=\left[\left( \mathbf{x}\cdot\mathbf{y}\right)^2-\left(\mathbf{x}\cdot\mathbf{x}\right)\left(\mathbf{y}\cdot\mathbf{y}\right)\right] \left(\mathbf{x}\times\mathbf{y}\right)\,.
\end{align}
\end{lemma}
The proof is similar to Lemma \ref{lemma:cross_product}.

We have proved the semi-positive definite property of $D$ by considering $\mathbf{q}^TD\mathbf{q}$ and completing the square. From equation \eqref{eq:d} and \eqref{eq:q}, we observe that points with $\dot{\psi}=\nabla\psi\cdot\dot{\mathbf{q}}=\nabla\psi\cdot\mathbf{f}=0$ may cause singularity. It can be proved that
\begin{theorem}[]
\label{theorem_limitset}
The union of all the $\omega$-limit sets \cite{hirsch1974differential} for solutions starting from all points $\mathbf{q}$ in the phase space is denoted by
$S=\bigcup_{\mathbf{q}} \omega(\mathbf{q})$,
then $\dot{\psi}(\mathbf{q})=0$ for all $\mathbf{q}\in S$.
\end{theorem}

\begin{proof}
We apply proof by contradiction.
If for a given point $\mathbf{q}_0$, $\dot{\psi}(\mathbf{q}_0)<0$. Since $\psi$ and $\dot{\mathbf{q}}=\mathbf{f}(\mathbf{q})$ are continuous differentiable functions then $\dot{\psi}=\nabla\psi\cdot\dot{\mathbf{q}}$ is continuous. So there exists an open set $\mathcal{O}$ as such for any $\mathbf{q}\in\mathcal{O}$, $\dot{\psi}(\mathbf{q})<c<0$ where c is a negative constant.

Suppose $\mathbf{q}_0\in S$, then there is a solution $\mathbf{q}(t)$ and a sequence $t_i$, $i\in\mathbb{N}$, $\mathbf{q}(t_i)\in \mathcal{O}$ and $\lim_{i\to \infty}\mathbf{q}(t_i)=\mathbf{q}_0$.
Since $\psi\in C^1$ and satisfies $\dot{\psi}\leqslant0$, the trajectory $L$ from $\mathbf{q}(t_i)$ to $\mathbf{q}(t_{i+1})$ must have a segment $\Delta L\subseteq L\cap\mathcal{O}$ inside $\mathcal{O}$, we have
\begin{align*}
\psi\left(\mathbf{q}(t_{i+1})\right)-\psi\left(\mathbf{q}(t_i)\right)&= \int_L\dot{\psi}dt=\int_{L\setminus\mathcal{O}}\dot{\psi}dt+\int_{L\cap\mathcal{O}}\dot{\psi}dt\\
&\leqslant c||\Delta L||<0\,.
\end{align*}
Hence $\{\psi({\mathbf{q}(t_n))}\}$ is a strictly monotonically decreasing sequence. By the continuity of $\psi$, $\lim_{i\to \infty}\psi(\mathbf{q}(t_i))=\psi(\mathbf{q}_0)$,
thus
\begin{align}
\label{eq:contradict}
\psi(\mathbf{q}(t_i))>\psi(\mathbf{q}_0), i\in\mathbb{N}\,.
\end{align}

Let $\mathbf{q}_0(t)$ be the solution starting at $\mathbf{q}_0$. For any $s>0$, we have $\psi(\mathbf{q}_0(s))<\psi(\mathbf{q}_0)$. Since $\mathbf{f}\in C^1$, then the solutions are continuously dependent on initial conditions \cite{hirsch1974differential}. Therefore, for any solution $\mathbf{q}'(t)$ starting sufficiently near $\mathbf{q}_0$, $\exists \varepsilon>0$
\begin{align*}
\psi(\mathbf{q}'(s))-\frac{\varepsilon}{2}<\psi(\mathbf{q}_0(s))<\psi(\mathbf{q}_0)-\varepsilon\Rightarrow \psi(\mathbf{q}'(s))<\psi(\mathbf{q}_0)\,.
\end{align*}
As a result, there exists adequately large $n$
\begin{align}
\psi(\mathbf{q}(t_n+s))<\psi(\mathbf{q}_0)\,,
\end{align}
which conflicts with \eqref{eq:contradict}. Hence $\mathbf{q}_0\not\in S$.
\end{proof}
Discussions on singularity for fixed points and limit cycles are presented below. Such singularity is canceled sometimes by the numerator in \eqref{eq:d} and \eqref{eq:q} but unavoidable for other cases where it exists, reflecting the nature of the dynamics.
\begin{itemize}
\item
All fixed points $\mathbf{q}^*\in S$, for $\mathbf{q}(t)=\mathbf{q}^*$ is a trivial solution, thus leading to singularity of $D$. Since $\nabla \psi(\mathbf{q}^*)=0$ and $\dot{\mathbf{q}}(\mathbf{q}^*)=0$, this type of singularity will not impact the potential function $\psi$.
\item
Points on limit cycles belong to $S$ where $\dot{\psi}=0$. Moreover, $\psi$ has to be equal-potential everywhere on a limit cycle fulfilling (\ref{item_g:b}) of the Definition \ref{def:cand_Lyapunov}. Together with the dissipation nature, we obtain $\nabla\psi=0$ \cite{zhu2006limit}. This demonstrates that on a limit cycle, the system is no longer dissipative but conserved.
\end{itemize}

Intuitively, a global Lyapunov function for a deterministic dynamical system will satisfy $\dot{\psi}\leqslant0$, which indicates all possible places in the phase space when $t\to \infty$ as the evolution converges.
Therefore, such a function with $\dot{\psi}<0$ except an unavoidable set of singular points (e.g.~points in $S=\bigcup_{\mathbf{q}} \omega(\mathbf{q})$) is usually what we seek, providing
 a strongest prediction on the evolution result.

Note that the construction is not unique. As the formerly presented construction starts from \eqref{eq:CFE}, we can symmetrically provide another one from \eqref{eq:SFE}. Since $[D+Q]\nabla\psi=-\dot{\mathbf{q}}$, then based on Lemma 1, $D=-\frac{\mathbf{f}\cdot\nabla\psi}{\nabla\psi\cdot\nabla\psi}E$, $Q=-\frac{\mathbf{f}\times\nabla\psi}{\nabla\psi\cdot\nabla\psi}$, $S=-\left[\frac{\nabla\psi\cdot\nabla\psi}{\mathbf{f}\cdot \nabla \psi}
E+\frac{\left(\mathbf{f}\times \nabla \psi\right)^2}{\left(\mathbf{f}\cdot \nabla \psi \right)\left(\mathbf{f}\cdot\mathbf{f}\right)}\right]$ and
$T=\frac{\mathbf{f}\times \nabla \psi}{\mathbf{f}\cdot\mathbf{f}}$. The explicit constructions presented in this letter have not been discovered from similar frameworks such as \cite{wang2003generalized}. These different constructions will not affect our equivalence result.

\begin{example}[Hamiltonian System]
For Hamiltonian systems \cite{arnold1989mathematical},
\begin{align*}
\dot{\mathbf{q}}
=\left(\frac{\partial H}{\partial p_1},...,\frac{\partial H}{\partial p_n},-\frac{\partial H}{\partial q_1},...,-\frac{\partial H}{\partial q_n}\right)^T
=-J\nabla H\,,
\end{align*}
where $J=\left(\begin{array}{cc}0 &-E\\E &0\end{array}\right)$. Thus it is apparent that \eqref{eq:SFE} is satisfied with $D=0$ and $\psi=H$, for $J$ is antisymmetric. Then by letting $T=-J$ and $S=0$ for energy conservation,
\begin{align*}
\left[S+T\right] \dot{\mathbf{q}}=T\dot{\mathbf{q}}=-J\dot{\mathbf{q}}&=-\nabla H=-\nabla \psi \Rightarrow \psi=H\,.
\end{align*}
\end{example}
The global Lyapunov function is actually the Hamiltonian. Therefore, the generalization of the vector-valued cross product in \eqref{eq:CFE} is appropriate. Besides, all systems discussed in \cite{PhysRevLett.81.2399} can be decomposed uniformly by this approach.

\begin{example}[Saddle Point System]
The system
\begin{equation*}
   \left\{
     \begin{array}{l}
        \dot{x}= x\\
        \dot{y}= -y
      \end{array}
   \right.
\end{equation*}
can be rewritten as
\begin{align*}
\left[S+T\right]\dot{\mathbf{q}}&=
\left(\begin{array}{cc}
1 & 1 \\
-1 & 1
\end{array}\right)\left(\begin{array}{c}
 x\\
 -y
 \end{array}\right)
=-\left(\begin{array}{c}
 -x+y\\
 x+y
 \end{array}\right) =-\nabla\psi\,,
 \end{align*}
 with $\psi(x,y)=-\frac{1}{2}x^2+xy+\frac{1}{2}y^2$.
One can check that
\begin{align*}
-\left[D+Q\right]\nabla \psi=-\frac{1}{2}\left(\begin{array}{cc}
1 & -1 \\
1 & 1
\end{array}\right)\left(\begin{array}{c}
 -x+y\\
 x+y
 \end{array}\right)
 =\dot{\mathbf{q}}\,
\end{align*}
where $D=\frac{1}{2}\left(\begin{array}{cc}
1 & 0 \\
0 & 1
\end{array}\right)$ and $Q=\frac{1}{2}\left(\begin{array}{cc}
0 & -1 \\
1 & 0
\end{array}\right)$.
Thus $\psi(x,y)$ is a global Lyapunov function for the system.
\end{example}

The \textit{Lyapunov equation} points another issue of our framework --- the global Lyapunov function is usually non-unique for deterministic systems \cite{PhysRevLett.81.2399}; but a unique and quantitative measure once specifying the diffusion matrix $D$ \cite{ao2004potential}.
For linear systems $\dot{\mathbf{q}}=A\mathbf{q}$, if any two eigenvalues $\lambda_i$ and $\lambda_j$ of $A$ satisfy $\lambda_i+\lambda_j\neq 0$ then for any symmetric and positive definite matrix $R$, there is a unique symmetric and invertible matrix $P$ fulfilling the Lyapunov equation \cite{haddad2008nonlinear}
\begin{align}
\label{eq:lya_eq}
A^T P+PA+R=0\,.
\end{align}
The system then has a global Lyapunov function $L=\mathbf{q}^TP\mathbf{q}$.
If $P$ is also positive definite, $L$ will be reduced to a classical global strong Lyapunov function.
Under such a configuration, we observe $[S+T]=-2PA^{-1}$, $S=-PA^{-1}-\left(A^T\right)^{-1}P$, since $A$ has no zero eigenvalue.
By setting $D=\frac{1}{4}P^{-1}RP^{-1}$, it is straightforward to prove that there is a one-to-one correspondence between symmetric and positive definite matrices $R$ and $D$.
Hence, \eqref{eq:lya_eq} is in fact the generalized Einstein relation \eqref{GER} for these linear systems. This example clearly demonstrates that
there are usually more than one global Lyapunov function for a system with different $D$. For a certain $D$, \eqref{GER} will guarantee the uniqueness of the global Lyapunov function.

Ordinary differential equations for a deterministic dissipative system indicate only the property of dissipation along the trajectories without specifying the speed of dissipation. Thus, arbitrary speed is acceptable, leading to different global Lyapunov functions. Therefore, the global Lyapunov function obtained is merely a qualitative measure (a partial order, e.g.~the function value of different stable fixed points are not comparable) but not quantitative. By indicating the diffusion matrix $D$ which is in fact the microscopic description of dissipation, the details of the frictional force are provided. Hence, the speed of dissipation along the trajectories is provided, defining the global Lyapunov function as a unique quantitative measure. This measure can be expressed as a Boltzmann-Gibbs distribution \cite{ao2008emerging} on the final steady state (if steady state exists) of the system's evolution.

In conclusion, we have presented a global Lyapunov function that coexists with complex dynamical behaviors as the natural generalization of the classical Lyapunov function. We have provided a constructive proof on the equivalence of the global Lyapunov function and the potential function obtained through a physical treatment of general dynamics. This relationship suggests new approaches on the construction of Lyapunov functions. Finally, we point out that for linear systems, the Lyapunov equation is a reduced form of the generalized Einstein relation.

\begin{acknowledgments}
This work was supported in part by the China 985 Initiative via Shanghai Jiao Tong University (R.Y., Y.M., B.Y. and P.A.); by the National 973 Projects No.~2007CB914700 and No.~2010CB529200 (P.A.); and by the Chinese Natural Science Foundation No.~NFSC61073087 (R.Y., Y.M. and B.Y.).
\end{acknowledgments}

\end{document}